\documentclass[onecolumn]{article}

\usepackage[margin=2.0cm,top=2.0cm]{geometry}

\usepackage{amsmath}
\usepackage{amsthm}
\usepackage{graphicx}
\newtheorem{theorem}{Theorem}[section]
\newtheorem{corollary}{Corollary}[theorem]

\newcommand{\keywords}[1]{\textbf{\textit{Keywords:}} #1}

\begin{document}

\title{An improvement of degree-based hashing (DBH) graph partition method, using a novel metric}

\author{Anna Mastikhina \and Oleg Senkevich \and Dmitry Sirotkin \and Danila Demin \and Stanislav Moiseev} \date{Huawei Technologies Co., Ltd.\\
	 \today}

\maketitle
\begin{abstract}
  This paper examines the graph partition problem and introduces a new metric, MSIDS (maximal sum of inner degrees squared). We establish its connection to the replication factor (RF) optimization, which has been the main focus of theoretical work in this field. Additionally, we propose a new partition algorithm, DBH-X, based on the DBH partitioner. We demonstrate that DBH-X significantly improves both the RF and MSIDS, compared to the baseline DBH algorithm. In addition, we provide test results that show the runtime acceleration of GraphX-based PageRank and Label propagation algorithms.
\end{abstract}

\keywords{Distributed Graph Computing, Graph Partitioning, Degree-Based Hashing, Power Law, Large-scale}

\section{Introduction}

In recent years, big graphs have become prevalent in various real-world applications such as web and social network services. Additionally, graphs are widely used for machine learning and data mining algorithms. The entire graph must be partitioned across the machines to execute distributed graph computing on clusters with multiple machines or servers. Graph partitioning quality in this case significantly impacts the performance of any algorithm for this graph. The essential metrics, which have been extensively researched previously are the workload balance and communication cost. So the ideal graph partition method should minimize the cost of cross-machine communication and maintain an approximately balanced workload across all processes.

Let's consider a graph $G = (V, E)$, where $V=\{v_i\}, i=\overline{1,n}$ represents the set of vertices, and $E \subseteq V \times V$ represents the set of edges in $G$. The cardinality of $V$ is denoted as $|V|$, which is equal to $n$. If there exists an edge $(v_i, v_j) \in E$, then $v_i$ and $v_j$ are said to be neighbors. The degree of the vertex $v_i$ is defined as the number of its neighbors and is denoted by $d_i$. Here we will consider directed graphs, the partition method can be also applied to undirected graphs.

There are two main directions of graph partitioning: so-called vertex-cut and edge-cut.
Suppose we have a cluster consisting of $m$ partitions. The goal of vertex-cut graph partitioning is to distribute each edge, along with its corresponding vertices, to one of the $m$ partitions in the cluster. It is important to note that each edge can only be assigned to one partition, while vertices can be duplicated across different partitions. For the case of edge-cut, the objective is to distribute the vertices among the partitions, allowing for possible duplicates of edges.

The edge-cut partition was initially a widely used approach. Edge-cut minimization is effective in computational grids and was used for solving differential equations. However, finding the optimal solution for minimizing edge-cut in a graph is an NP-hard problem. Therefore, heuristic approaches have been utilized since 1970s \cite{klin}. However, vertex partition has been found to be ineffective for real-world graphs with skewed distributions, such as power-law distributions \cite{powerlaw}, \cite{powergraph}. As a result, the focus of this paper will be on the vertex-cut approach.

The edge partition technique has been found to be more suitable for real-world graphs with irregular structures compared to the vertex partition approach. This is particularly true for social graphs, which often exhibit a power-law distribution where the number of vertices with a certain degree decreases exponentially as the degree increases \cite{powerlaw}.

There are several graph processing systems available that not only offer the standard set of graph analytics algorithms but also incorporate built-in partitioning algorithms. Some of these systems include the PowerGraph \cite{powergraph},  PowerLyra \cite{powerlyra}, GraphX \cite{graphx}, and the Pregel \cite{graphxcost} family.

In this paper, we will use GraphX for partition quality evaluation. GraphX is a distributed graph engine that is developed using Spark \cite{graphx}. It enhances Resilient Distributed Graphs (RDGs), which are used to associate records with vertices and edges. GraphX uses a flexible vertex-cut partitioning approach to encode graphs as horizontally partitioned collections by default. The number of partitions is equivalent to the number of blocks (64 Mb) in the input file. GraphX is widely used for graph analytics both in research and the industry. More specifically, we will use algorithms based on a GraphX Pregel API for the method evaluation \cite{graphxcost}. 

\subsection{Article structure}

The article is structured as follows: Chapter \ref{metrics_chap} will discuss the metrics that affect performance time, including a novel metric. Chapter \ref{sota} will provide a brief overview of the state-of-the-art partition algorithms, with a focus on two baseline methods: EdgePartition2D and DBH. In Chapter \ref{dbhx}, we will introduce our new method, called DBH-X. Theoretical analysis will be presented in Chapter \ref{theory}, including a theorem on the relationship between partition metrics. Chapter \ref{results} will present the testing results of our implementation of DBH-X, including an examination of acceleration of the algorithms used in GraphX. The last chapter \ref{conclusion} is a conclusion.
\section{Metrics for the edge partitioning}\label{metrics_chap}

\subsection{Standard metrics for the Vertex-Cut partitioning}

We consider edge partitioning with the sets of partitions $E_i$, where $E = E_1 \cup E_2 \cup \ldots \cup E_m$. The following metrics have been used in the literature for quality control of the partitioning:

\begin{enumerate}
	\item Balance
	\begin{equation}\label{eq:balance}
		B = \frac{\max_{i} |E_i|}{\frac{|E|}{m}}
	\end{equation}
	
	\item Replication Factor (RF)
	\begin{equation}\label{eq:rf}
		RF=\sum_{i=1}^{m} \frac{|V(E_i )|}{|V|}  
	\end{equation}

	where $V(E_i)$ is a set of vertices incident to edges of $E_i$,
	
	\item Number of repeated vertices in different partitions 
	
	\begin{equation}\label{eq:r}
		R= \sum_{i=1}^m (|V(E_i )|)-|V| 
	\end{equation}
	
	\item 	The number of frontier vertices --- vertices appearing in more than one partition
	
	\item Communication cost   
	\begin{equation}\label{eq:communication_cost}
		C= \sum_{i=1}^m |F(E_i )|
	\end{equation}
	
	where set $F(E_i)$ consists of the frontier vertices adjacent to edges from $E_i$.

 A better partition should have smaller values for these metrics.
	
\end{enumerate}

\subsection{Maximal Sum of Inner Degrees Squared}

We introduce a novel metric for evaluating partition quality called the Maximal Sum of Inner Degrees Squared on one partition (MSIDS). It can be computed as follows:

\begin{equation}\label{eq:msids}
	MSIDS = \max_j \sum_{v\in |V(E_j)|} d(j, v)^2,
\end{equation}

where $d(j,v)$ denotes inner degree of the vertex $v$ located on some partition $E_j$ which means degree for graph $(V(E_j), E_j)$.

Typically, optimal performance can be achieved by using graph partitioning with metrics such as balance and replication factor close to their optimal values. However, there are cases where these metrics do not accurately reflect algorithm performance time. For example, the Label Propagation algorithm implemented in GraphX's Pregel may experience slow performance on the first iteration when using a partition with good balance and replication factor, whereas a partition with slightly worse metric values may result in better runtime. A similar issue arises with the shortest path algorithm.

In Table \ref{tab:metr} we provide metrics (1) --- (5) for two partitions obtained by using EdgePartiton2D method, implemented in GraphX (denoted as Edge2D), and a method where all edges are distributed evenly to partitions but keeping edges with the same source vertices on one partition (denoted as Partition1). Partition1 has better values for almost all metrics 1)-5) described in the previous subsection. As a result, Connected Components algorithm works faster. But the runtime of LPA with this partition is more than 50 minutes while the runtime with EdgePartition2D is around 6 minutes.

\begin{table}
\caption{\label{tab:metr} different metrics for 2 partition methods ( graph uk-2002, number of partitions is 210)}
	\centering
	\begin{tabular}{c|c|c|c|c|c|}
		partition  & balance & RF & \# repeated & \# frontier & communication \\
		algorithm & &  & vertices & vertices & cost \\\hline
		Edge2D & 1.0102 & 16.17 & $2.8\cdot 10^8$ & $18.5\cdot 10^6$ & $20.9\cdot 10^7$ \\\hline
		Partition1 & 1.0155 & 1.47 & $8.7\cdot 10^6$ & $4.3\cdot 10^6$ & $1.3\cdot 10^7$
	\end{tabular}
\end{table}

In these cases, message merging is a bottleneck. Consider Pregel implementation of the Label propagation algorithm using Spark version 2.4.5 and Scala version 2.11.8. Messages are formed for all edges of the partition and are sent to both source and destination vertices. The {\bf mergeMsg} function is then used to combine all messages destined to the same vertex.

\begin{equation*}
	\begin{split}
def\ {\bf mergeMessage}(
	count1: Map[VertexId, Long],\\
 	count2: Map[VertexId, Long]):\\
 	Map[VertexId, Long] = \{\\
	(count1.keySet ++ count2.keySet).map \{ i =>\\
		val\ count1Val = count1.getOrElse(i, 0L)\\
		val\ count2Val = count2.getOrElse(i, 0L)\\
		i -> (count1Val + count2Val)\\
	\}(collection.breakOut)
\end{split}
\end{equation*}

This process occurs within a single partition without interaction with others. Number of edges inside one partition $E_i$ is  $(\sum_{v\in V(E_i)} deg_{inner} (v))/2$.
Message merging is performed consequently, in this implementation for one vertex we perform $d$ (inner degree) summations of $(set, value)$ for each neighbor, where the set has the size from $1$ to $d$. So
merging all messages has the number of operations depending on $deg_{inner} (v)^2$.
That makes dangerous high inner vertex degrees inside one partition, this stage can be significantly slowed down.

Therefore, high inner vertex degrees within a partition can be detrimental. In Chapter \ref{metr_restr}, we will analyze the relationship between the replication factor and the maximal sum of inner degrees squared.

\section{Existing approaches}\label{sota}

Despite numerous studies on the graph partitioning problem, there is no one-size-fits-all method that can accelerate any algorithm for any graph. Generally, in edge partitioning, most researchers agree on the importance of reducing the replication factor (i.e., the average number of vertex replicas) while maintaining a good balance to improve runtime.

Partition algorithms can be broadly categorized into two types: those that use extensive graph information, which can be gathered for example by locally traversing some parts of the graph, and those that use less information, resulting in faster partitioning. Algorithms that utilize more graph information, like HEP \cite{hep} or NE \cite{ne}, tend to yield better acceleration but the partition itself takes a longer time. In contrast, other algorithms, use less information but can perform partitioning more quickly. In practice, the choice of partition algorithm depends on the algorithm being accelerated, as well as the size and structure of the graph. A tradeoff must be made between partition quality and time efficiency in each specific case.

It is important to note that the choice of partitioning algorithm depends on various factors such as graph size, structure, and the algorithm to be accelerated. For instance, if the graph is very large, a linear time algorithm like DBH \cite{dbh} may be more appropriate as it is faster than algorithms like METIS \cite{metis} or subgraph-based approach \cite{subgraph} that use more graph information. Similarly, if the algorithm being accelerated has a quadratic dependence on the number of inner degrees of vertices, then partitioning algorithms that tend to have large inner degrees for partitions (HEP \cite{hep}, NE \cite{ne}) may be less suitable.

Furthermore, some partitioning algorithms are designed to handle specific types of graphs, such as Powerlyra \cite{powerlyra} which is optimized for power-law graphs. In such cases, it may be beneficial to use specialized algorithms like DBH.

Finally, it is worth noting that the usage of streaming techniques can also affect the choice of partitioning algorithm. For example, some algorithms like HDRF \cite{hdrf} and AdWISE \cite{adwise} are designed to work well with streaming data. Therefore, when choosing a partitioning algorithm, one must consider the specific requirements and characteristics of the problem at hand.

In this study, we don't consider streaming techniques and algorithms requiring too large time for partition itself. The goal is to receive target algorithms acceleration in such a way that end-to-end time including graph partition time will be minimal. 

\subsection{EdgePartition2D}

One very common partition method is called EdgePartition2D. It maps the edges of the graph into a two-dimensional grid of $n$ partitions, where $n$ is a user-defined parameter. To assign an edge to a partition, EdgePartition2D computes the hash function for each endpoint of the edge. Then, the edge is assigned to the partition, corresponding to the grid cell that matches the pair of these hashes. This method guarantees the replication factor less than $2\sqrt{n}-1$, where $n$ is the number of partitions. 

More formally, let's consider the adjacency matrix of a given graph, which is a square matrix of size $|V|\times|V|$ where each element $(i,j)$ is 1 if there is an edge between vertices $v_i$ and $v_j$ and 0 otherwise. We can split this matrix into square blocks of equal size, resulting in a total of $n$ blocks arranged in $\sqrt{n}$ rows and $\sqrt{n}$ columns. Each block represents a partition, and the edges contained within that block belong to that partition. For example, the edge $(i,j)$ belongs to the partition corresponding to the block that contains the element $(i,j)$.

If the desired number of partitions is not a perfect square, then some partitions will be combined. Now let's consider a single vertex $i$. All edges adjacent to this vertex can only exist within the $i$-th row and column of the matrix. Thus, the replication factor for vertex $i$ is the number of partitions that contain edges adjacent to $i$. Therefore the replication factor for any vertex is bounded by $2\sqrt{n}-1$.

EdgePartition2D is frequently used as a reference partitioning method implemented in GraphX. The figure \ref{fig:Edge2D_ex}
 shows the graphical representation of EdgePartition2D.

\begin{figure}
	\centering
	\includegraphics[width=0.6\textwidth]{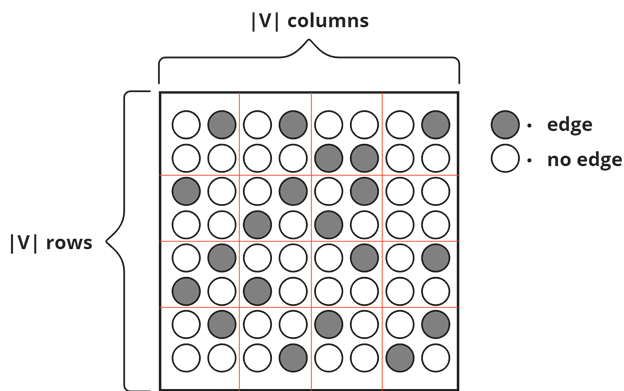}
	\caption{\label{fig:Edge2D_ex} Edge2D partition}
\end{figure}

\subsection{DBH partition}

Degree-based hashing (DBH) sets the partition ID to the edge $(v_i, v_j)$ where vertices $v_i, v_j$ have degrees $d_i, d_j$ by the following rule:

\begin{equation}
	edge\_hash(v_i,v_j) = \begin{cases}
		vertexHash(v_i) & \text{if $d_i< d_j$,}\\
		vertexHash(v_j) & \text{otherwise.}\\
	\end{cases}
\end{equation}

A commonly used method for generating the hash function is to take the remainder of the division of the vertex index by the total number of partitions.

\begin{equation}
vertexHash(v_i) = v_i  \mod numParts.
\end{equation}

The original article \cite{dbh} demonstrates that DBH produces partitions with a lower replication factor compared to randomly assigned edges while maintaining good balance.

\textbf{DBH partition example}

We consider a graph from figure \ref{fig:dbh} in this example. The partition number we set as 3. We denote as $P(u, v)$ a partition id for the edge $(u,v)$.

\begin{equation*} \label{eq1}
	\begin{split}
		deg(0) \textless deg(1), P(0, 1) & = 0 \\
		deg(0) \textless deg(3), P(0, 3) & = 0 \\
		deg(1) \textless deg(4), P(1, 4) & = 1\\
		deg(1) \geq deg(5), P(1, 5) & = 5\mod 3 = 2\\
		deg(2) \textless deg(0), P(2, 0) & = 2\\
		deg(2) \textless deg(3), P(2, 3) & = 2\\
		deg(3) \geq deg(4), P(3, 4) & = 4\mod 3 = 1\\
		deg(4) \textless deg(5), P(4, 5) & = 4\mod 3 = 1\\
		deg(5) \textless deg(3), P(5, 3) & = 5\mod 3 = 2
	\end{split}
\end{equation*}

\begin{figure}
	\centering
	\includegraphics[width=0.6\textwidth]{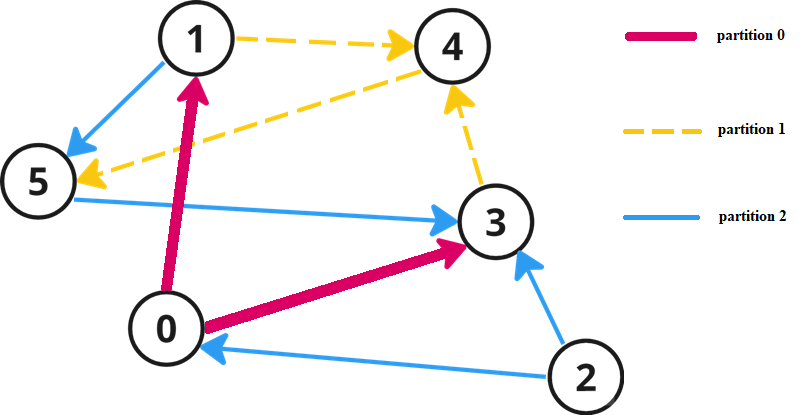}
	\caption{\label{fig:dbh} DBH partition, red lines --- edges on the partition 0, yellow --- partition 1, blue --- partition 2}
\end{figure}

Using formula \ref{eq:rf}, we can compute the replication factor for this case, which is $\frac{2+3+1+3+1+2}{6} = 2$. The theoretical background for the DBH method is provided in section \ref{sec:dbh_theory}

\section{DBH-X partition approach}\label{dbhx}
This section introduces two enhancements to the DBH method, which improve both metrics and performance time for the considered algorithms.

\subsection{Threshold $\tau$}

We can improve the DBH method by introducing a threshold $\tau$ for vertex degrees. In this method, the edge partition is defined as follows:

\begin{equation}
	\begin{gathered}
	If\ d_i>\tau\ or\ d_j>\tau:\\ 
	edge\_hash(v_i,v_j) = \begin{cases}
		vertexHash(v_i) & \text{if $d_i\leq d_j$}\\
		vertexHash(v_j) & \text{if $d_i > d_j$}\\
	\end{cases}, \\  
	\addtocounter{equation}{1}\tag{\theequation} 
	else:\\ 
 edge\_hash(v_i,v_j) = \begin{cases}
		vertexHash(v_i) & \text{if $v_i\leq v_j$}\\
		vertexHash(v_j) & \text{if $v_i > v_j$}\\
	\end{cases}     
\end{gathered}
\end{equation}

	This technique partitions certain edges based on the hash function derived from the source vertex ID. (We can think that vertices are numbered in a way  that the source vertex has lower number than the destination vertex.) Specifically, it partitions the subgraph induced by the low-degree vertices using the source partition instead of the DBH partition. As this subgraph is not significantly skewed, using the DBH method does not necessarily improve the replication factor in a provable way. In practice, threshold usage has been found to result in better replication factor values (see Table \ref{tab:rfs}).
	
	\subsection{Spread technique}
	
	In section \ref{metr_restr}, we will show that decreasing one metric, such as the replication factor, can result in an increase in another metric, such as the maximal sum of inner degrees squared. High inner degrees are a common occurrence in the DBH method. Due to this method, all the edges leading from some vertex to a vertex with a higher degree are placed in a single partition. As a result, a vast majority of edges leading from any vertex end up in the same partition, while the remaining edges are distributed somewhat evenly. To address this issue, we use a parameter called the "spread", it distributes such edges among different sets, leading to a more balanced distribution of edges across partitions.
	
	The spread parameter controls the distribution of partitions into several sets. Each edge is first distributed into one of the parts, and only afterward is distributed to one of the partitions from this part. This helps to control the spread of edges between partitions and improve the balance of the maximal sum of inner degrees squared.
	
	The spread parameter is used to control the maximal sum of inner degrees squared, which can increase when the replication factor is decreased. To accomplish this, the partitions are separated into multiple sets, and each edge is first assigned to a set and then to a partition within that set. Any hash function might be used for this --- the simplest one is just taking a division remainder:
	
	\begin{equation}\label{eq:partitionset}
		whichSet(v_i,v_j) = (v_i  +v_j )\mod spread
	\end{equation}
	
	The operation we have described can be seen as dividing the graph into multiple subgraphs, whose number is determined by the \textit{spread} parameter. Each of these subgraphs is then partitioned into $numParts/spread$ parts. This technique helps to avoid high inner degrees in the partition, which can improve computation time on a single partition and prevent computational bottlenecks. It should be noted that this technique still maintains a balanced distribution.
	
	\begin{figure}
		\centering
		\includegraphics[width=0.5\textwidth]{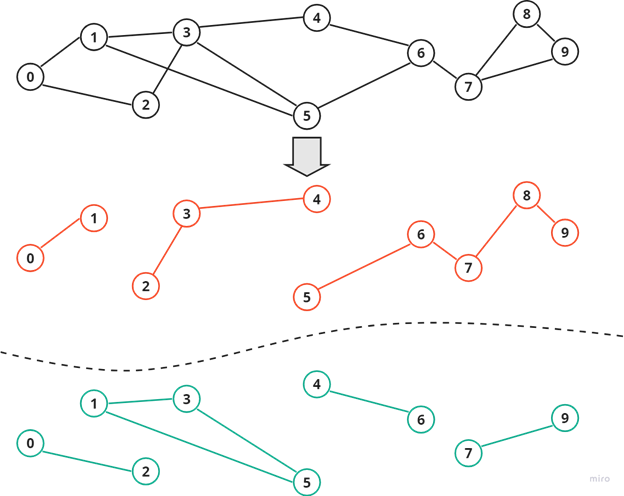}
		\caption{Using spread = 2}
		\label{fig:spread}
	\end{figure}

	\subsection{DBH-X algorithm}
	
	In this section, we describe the DBH-X algorithm, which uses the proposed techniques.
	
	We assume that the partition number can be divided by \textit{spread} without a remainder. If this is not the case, slight modifications are required, namely, some of the sets will contain one additional partition.
	
	\begin{equation}\label{eq:partsInSet}
		partsInSet = numParts / spread
	\end{equation}
	
	For each edge $(v_i, v_j)$ we define $whichSet(v_i,v_j)$ according to \ref{eq:partitionset}.
	
	In addition, we set a hashing function as a remainder from a division of vertex id by the number of partitions $m$: 
	
	\begin{equation}\label{eq:vertexHash}
		vertexHash(v_i, m) = i\mod m
	\end{equation}
	
	Let's consider edge $(v_i, v_j)$ with the corresponding vertices' degrees $d_i$ and $d_j$.
	
	\begin{equation}
		\begin{gathered}
If d_i>\tau\ or\ d_j>\tau:\\
		edge\_hash(v\_i,v\_j)
		=\begin{cases}
			vertexHash(v_i) + whichSet*partsInSet & \text{if $d_i\leq d_j$}\\
			vertexHash(v_j)  + whichSet*partsInSet & \text{if $d_i > d_j$}\\
		\end{cases},\\ 
		\addtocounter{equation}{1}\tag{\theequation} 
else:\\
edge\_hash(v_i,v_j)
		 = \begin{cases}
			vertexHash(v_i) + whichSet*partsInSet  & \text{if $v_i\leq v_j$}\\
			vertexHash(v_j) + whichSet*partsInSet  & \text{if $v_i > v_j$}\\
		\end{cases}   
\end{gathered}
	\end{equation}
	
	By using this partitioning technique instead of DBH, we may sacrifice some replication factor increases for the sake of avoiding high inner degrees within partitions.
	
	\section{Theoretical analysis}\label{theory}
	
	\subsection{DBH theoretical background}\label{sec:dbh_theory}
	
	Let's consider graph $G=(V,E)$ with a set of vertices $V=\{v_i\}_{i=1}^n$ and their corresponding degree sequence denoted by $D = \{d_i\}_{i=1}^n$. We consider a randomized edge distribution to $m$ partitions. The probability of an edge $\{v_i, u\}$ being assigned to the $j$-th partition is $1/m$, and the probability of it not being assigned to the $j$-th partition is $(1-1/m)$. The probability of a copy of vertex $v_i$ being present to the $j$-th partition is $(1-(1-1/m)^{d_i})$, where $d_i$ is the degree of $v_i$. The number of replicas of vertex $v_i$ is denoted by $A(v_i)$.
	
	The following 2 theorems from the article \cite{dbh} provide some theoretical background for the DBH partitioning algorithm.
	
	\begin{theorem} A randomized vertex-cut in $m$ partitions has the expected replication factor of
		
		\begin{equation}
			\mathbf{E} [\frac{1}{|V|} \sum_{v_i\in V}|A(v_i )| ]= \frac{m}{|V|} \sum_{v_i\in V}(1-(1-1/m)^{d_i})   
		\end{equation}
		
	\end{theorem}
	
	In many partitioning algorithms, edges are hashed based on the ID of one of their vertices. Let us consider such a partitioner. 
	
	Following \cite{dbh} we introduce a set $H=\{h_i\}_{i=1}^n$ where $h_i$ denotes the number of $v_i$’s adjacent edges which are hashed by the index of neighbors of $v_i$ according to the vertex-hash function. Additionally, we will refer to edges that are adjacent to vertex $v_i$ and are hashed by the ID of the neighbor of $v_i$ as "decentrally hashed".
	
	\begin{theorem} \label{th_dbh}
		DBH method for $m$ partitions has the following expected replication factor:
		\begin{equation*}
			\mathbf{E} [\frac{1}{|V|} \sum_{v_i\in V}|A(v_i )| ]=\frac{m}{|V|}  \sum_{v_i\in V}(1-(1-1/m)^{h_i+g_i }) \leq 
		\end{equation*}
		\begin{equation}
			\leq \frac{m}{|V|}  \sum_{v_i\in V}(1-(1-1/m)^{d_i}),
		\end{equation}
		where 
		\begin{equation*}
			g_i=\begin{cases}
				1 & \text{if $h_i< d_i$,} \\
				0 & \text{if $h_i=d_i$} \\
			\end{cases}
		\end{equation*}
		
	\end{theorem}
	
	In Theorem \ref{th_dbh}, we worked with sets of decentrally hashed edges, although the theorem did not specify the exact method employed. So, it is important to note that this evaluation is not specific to the DBH partitioner and can also be applied to similar approaches. For exactly one of the ending vertices of the edge, an edge will be decentrally hashed. Therefore, $\sum_{i=1}^n h_i = \frac{1}{2} \sum_{i=1}^n d_i$.
	
	To evaluate the replication factor more precisely for power-law graphs, first, we'll provide some results on replication factor evaluation obtained by Xie et al. in paper \cite{dbh} --- the original authors of the DBH method.


	\begin{equation}\label{comparison_05}
		\begin{gathered}
		\frac{m}{|V|}  \sum_{v_i\in V}(1-(1-1/m)^{h_i+g_i })\leq
		 \frac{m}{|V|}  \sum_{v_i\in V}(1-(1-1/m)^{d_i}),
		 	\end{gathered}
	\end{equation}
	
	
	
	In addition, the authors in \cite{dbh} suggest that the optimal distribution of $h_i$ values, for minimizing the replication factor, is achieved when as many values as possible are equal to the greatest $d_i$, while the remaining values are set to zero.
	
	To evaluate $\mathbf{E}[h_i+g_i]$ we denote $b_i$ --- the probability that the random edge adjacent to $i$-th vertex is decentrally hashed. Then $\mathbf{E}[h_i+g_i] = b_i\cdot d_i +1 - b_i^{d_i}\leq b_i\cdot d_i + 1$. For the DBH algorithm $b_i = P(d < d_i) + \frac{1}{2}\cdot P(d = d_i)$ by definition.
	
	So, for replication factor maximization, $b_i$ should be close to $1$ for some $k$ vertices with the highest degrees such that $\sum_{i=1}^k d_i = \frac{1}{2} \sum_{i=1}^n d_i$ or close to this value. Other $b_i$ should be equal to $0$.
	
	For power-law graphs the number of vertices with degree $d$ is proportional to $d^{-\alpha}$ for some positive constant $\alpha$. 
	
	Let us denote the following function --- $B(\alpha, d_{min}, d_{max}) = \sum_{x=d_{min}}^{d_{max}} x^{-\alpha}$. 
	Then $b_i$ for DBH can be evaluated as the fraction of vertices with degrees less than $d_i$:
	
	\begin{equation}{\label{bi}}
		\frac{B(\alpha, d_{min}, d_i - 1)}{B(\alpha, d_{min}, d_{max})} < b_i < \frac{B(\alpha, d_{min}, d_i)}{B(\alpha, d_{min}, d_{max})}
	\end{equation}
	
	Xie et al. in their paper \cite{dbh} prove the following using the inequalities for integral sum.
	
	\begin{equation}
		\begin{gathered}
		\frac{d_{min}^{1-\alpha} - (d_{max} + 1)^{1-\alpha}}{\alpha - 1} < B(\alpha, d_{min}, d_{max}) <\frac{d_{min}^{1-\alpha} - d_{max}^{1-\alpha}}{\alpha - 1} + d_{min}^{-\alpha}
		\end{gathered}
	\end{equation}

	First, consider the upper bound for $b_i$.
	
	\begin{equation}{\label{bi_ineq}}
		\begin{gathered}
		b_i < \frac{B(\alpha, d_{min}, d_i)}{B(\alpha, d_{min}, d_{max})} < \frac{\frac{d_{min}^{1-\alpha} - d_i^{1-\alpha}}{\alpha - 1} + d_{min}^{-\alpha}}{\frac{d_{min}^{1-\alpha} - (d_{max} + 1)^{1-\alpha}}{\alpha - 1}} = \frac{d_{min}^{1-\alpha} - d_i^{1-\alpha} + (\alpha - 1)d_{min}^{-\alpha}}{d_{min}^{1-\alpha} - (d_{max} + 1)^{1-\alpha}}
		\end{gathered}
	\end{equation}
	
	The maximal degree $d_{max}\leq n-1$ though in practice it is much lower.
	$d_i << d_{max}$ for the most of $d_i$ for the power-law graphs and $d_{min}\leq d_i \leq d_{max}$ for all $i$.
	
	Xie et al. in \cite{dbh} take $d_{max} = n-1$ in \ref{bi_ineq} and make the following conclusions:

	\begin{enumerate}
		\item for $0 < \alpha < 1$ and $d_i << n$, $\lim_{n\to\infty} b_i = 0$,
		which means that low-degree vertices will be hashed centrally which is close to "ideal distribution",
		\item for $\alpha > 1$ upper bound of $b_i$ increases as $\alpha$ increases.
	\end{enumerate}
	
	In the first case, when $0 < \alpha < 1$, even if $d_{max} < n-1$ but $d_i$ is much smaller than $d_{max}$, $b_i$ decreases as $\alpha$ decreases. Consequently, for small degrees, $b_i$ approaches zero if $\alpha$ is small enough.
	
	When $\alpha > 1$, the upper bound of $b_i$ increases and can easily reach $1$. For instance, if $d_{min}=1$, $\alpha=1.6$, and $d_i \geq 3$, the numerator of the fraction will be greater than or equal to $1.6 - 3^{-0.6} > 1.08$, while the denominator is positive and less than $1$.
	
	In situations where the upper bound probability exceeds 1, it is important to consider the lower bound. This will be explored in the following subsection.
	
	\subsection{DBH and DBH-X comparison}
	
	In this subsection, we compare the standard DBH method with DBH-X, which involves a threshold parameter denoted as $\tau$. By introducing the threshold, we divide all vertices into two groups: high-degree vertices ($d_i > \tau$) and low-degree vertices ($d_i \leq \tau$).
	
	For high-degree vertices, the partitioning of their incidental edges remains the same as in the standard DBH method, and their contribution to the resulting replication factor is the same as in Theorem \ref{th_dbh}.
	
	However, the main difference arises for edges incidental to low-degree vertices, particularly edges that connect two vertices of low degree (low-to-low edges).
	
	The replication factor of the DBH method can be calculated by the following sum:
	
	\begin{equation} \label{eq:DBH}
		\begin{gathered}
		E_{DBH}[RF | D] = \frac{m}{|V|} (\sum_{v\in V_{high}} (1-(1-1/m)^{h_i+g_i})
		 + \sum_{v\in V_{low}} (1-(1-1/m)^{h_i+g_i}))
		\end{gathered}
	\end{equation}
	
	We use the notation $\overline{h_i}$ to represent the number of edges that are adjacent to vertex $v_i$ and are hashed decentrally in DBH-X.
	
	Then:
	
	\begin{equation}\label{eq:DBH-X}
		\begin{gathered}
		E_{DBH-X} [RF | D] = \frac{m}{|V|} ( \sum_{v\in V_{high}} (1-(1-1/m)^{\overline{h_i}+g_i}) 
		 + \sum_{v\in V_{low}} (1-(1-1/m)^{\overline{h_i}+g_i}))
		\end{gathered}
	\end{equation}
	
	For $v_i \in V_{high}$ the equality $h_i = \overline{h_i}$ holds. So, the only difference between equations \ref{eq:DBH} and \ref{eq:DBH-X} appears in the second term of the sum. Note that for the second terms of equations \ref{eq:DBH} and \ref{eq:DBH-X} the $h_i$ equals the number of $v_i$'s neighbors with the less or equal degree and $\overline{h_i}$ equals the number of $v_i$'s neighbors with lesser ID having low degree (low-to-high edges will be hashed centrally). Also, we note that if $v_i \in V_{low}$ then $h_i \leq \tau$ and $\overline{h_i} \leq \tau$, since $d_i \leq \tau$.
	
	Let us take a look at the lower bound of the inequality \ref{bi} for DBH:
	
	\begin{equation}\label{general_eq_bi}
	b_i > \frac{B(\alpha, d_{min}, d_i - 1)}{B(\alpha, d_{min}, d_{max})} > \frac{d_{min}^{1-\alpha} - d_i^{1-\alpha}}{d_{min}^{1-\alpha} - d_{max}^{1-\alpha} + (\alpha - 1)d_{min}^{-\alpha}}
	\end{equation}
	
	With $d_{min}=1$,
	
	\begin{equation}
	b_i > \frac{1-d_i^{1-\alpha}}{\alpha - d_{max}^{1-\alpha}}
    \end{equation}
	
	If $\alpha > 1$
	
	\begin{equation}\label{lower}
	b_i > \frac{1-d_i^{1-\alpha}}{\alpha}
	\end{equation}
	
	As $d_i$ increases, this lower bound of the probability \ref{lower} tends asymptotically to $1/\alpha$.
	
	\begin{corollary}
		Consider a power-law graph with $\alpha > 1$. $d_{min} \geq 1$. Then the probability $b_i$ that the random edge adjacent to the $i$-th vertex is decentrally hashed in the DBH partition method can be evaluated as
		$$
			b_i > \frac{1-(\frac{d_i}{d_{min}})^{1-\alpha}}{\alpha}.
		$$

	\end{corollary}
	
	\begin{proof}
		Take the right part of \ref{general_eq_bi} and divide the numerator and denominator by $d_{min}^{1-\alpha}$.
		
			\begin{equation}
		\frac{1 - (\frac{d_i}{d_{min}})^{1-\alpha}}{1 - (\frac{d_{max}}{d_{min}})^{1-\alpha} + (\alpha - 1)d_{min}^{-1}}
			\end{equation}
		
		Since
		$\frac{\alpha - 1}{d_{min}} < \alpha - 1$,
		
			\begin{equation}
		b_i > \frac{1 - (\frac{d_i}{d_{min}})^{1-\alpha}}{1 + (\alpha - 1)} =\\
			\frac{1 - (\frac{d_i}{d_{min}})^{1-\alpha}}{\alpha}
			\end{equation}
		
	\end{proof}

Comparing this value with $0.5$, we'll get the following simple corollary.

\begin{corollary}
	For $1<\alpha<2$ and $d_i > d_{min}\cdot(1-\frac{\alpha}{2})^{\frac{-1}{\alpha - 1}}$ for the coefficient $b_i$ denoting the probability that random edge adjacent to the $i$-th vertex is decentrally hashed in DBH partition method, the following inequality holds:
	 $$b_i> 0.5$$. 
\end{corollary}

    For example, take $\alpha = 1.6$, than for $\frac{d_i}{d_{min}} > 14$ the value $b_i > 0.5$.
    
    So, 
    for a rather large part of the vertices, the expectation is greater than $0.5\cdot d_i$, so it makes sense to use for this part a model with $b_i=0.5$ instead.
    
	In conclusion, there is no guarantee that hashing by degree distribution will always be optimal for all power-law graphs. 
	For some sets of high-degree vertices values of $b_i$ are rather big and can be considered as close to an "ideal" distribution. But for lower vertices "ideal" $b_i$ should be close to zero but they are greater than $0.5$ starting from some degree. Alternatively, if we choose to hash decentrally in a random way, we can set $b_i = 0.5$ approximately, which is preferable for a significant number of medium-degree vertices. Thus, we can utilize an adaptive threshold to categorize high-degree vertices and apply standard DBH to them while using hashing by vertex ID for low-degree vertices.
	
	\subsection{Metrics restriction} \label{metr_restr}
	
	Let us denote a replication factor of a vertex $v$ as $\rho(v)$. The replication factor for the whole graph equals 
	
	\begin{equation}
	RF=\frac{1}{|V|}  \sum_{v\in V} \rho (v).
	\end{equation}
	
	The degree of a vertex $v$ will be denoted as $d(v)$. The inner degree $d(i, v)$ represents the number of edges from partition $E_i$ that are incident to vertex $v$. The maximal sum of inner degrees squared (MSIDS) can be computed as follows:
	
	\begin{equation}
	MSIDS=\max_i \sum_{v\in V} d(i,v)^2
	\end{equation}
	
	Note that the above definitions of the replication factor and MSIDS don't conflict with definitions \ref{eq:rf} and \ref{eq:msids}.
	
	\begin{theorem}\label{inequality}
		
		For any graph $G=(V, E)$ and any partition into $m$ parts, it is impossible to simultaneously minimize the replication factor and the maximal sum of inner degrees squared within some limit. Specifically, the following inequality holds:
		
		$$
		RF \cdot MSIDS \geq \frac{4|E|^2}{m|V|} 
		$$
	\end{theorem}
	
	\begin{proof}
		First, let us note that by definition, $MSIDS \geq \sum_{v\in V} d(i,v)^2$ for every partition $i$. Taking the sum of these inequalities, we obtain:
		
	\begin{equation}
		m \cdot MSIDS \geq \sum_{i=1}^m \sum_{v\in V} d (i,v)^2
	\end{equation}
		
		Let's change the order of summation in the right part to get:
		
	\begin{equation}
		\sum_{i=1}^m \sum_{v\in V} d (i,v)^2 = \sum_{v\in V} \sum_{i=1}^m d (i,v)^2 
	\end{equation}
	
		It is important to note that $\rho(v)$ represents the exact number of partitions that contain vertex $v$. Thus, if vertex $v$ is not present on the $i$-th partition, then $d(i,v)=0$. Therefore, in the summation, the number of non-zero elements is equal to $\rho(v)$ instead of $m$.
		
		Let us use the inequality of arithmetic and quadratic means:
		
		\begin{equation*}
		\begin{gathered}
		\sum_{v\in V} \sum_{i=1}^m d (i,v)^2 = \sum_{v\in V} \rho (v) \frac{\sum_{i=1}^m d (i,v)^2}{\rho(v)} \geq 
		\sum_{v\in V} \rho (v) (\frac{\sum_{i=1}^m d(i,v)}{\rho (v) })^2 = 
		\sum_{v\in V} \frac{d(v)^2} {\rho (v)}
		\end{gathered}
	\end{equation*}
		
		This leads to inequality
		
		\begin{equation*}
		RF \cdot MSIDS \geq \rho \frac{1}{m} \sum_{v\in V}\frac{d(v)^2}{\rho (v)} = 
		\frac{1}{m|V|}  (\sum_{v\in V}\rho (v))\sum_{v\in V} \frac{d(v)^2}{\rho (v) }
		\end{equation*}
		
		Cauchy-Bunyakovsky-Schwarz inequality for sets of numbers 
		$\sqrt{\rho (v)}$ and $\frac{d(v)}{\sqrt{\rho (v)}} $ for every $v$ implies that
		\begin{equation*}
		(\sum_{v\in V}\rho (v))\sum_{v\in V}\frac{d(v)^2}{\rho (v)}\geq (\sum_{v\in V} d(v))^2=4|E|^2
	\end{equation*}
		
		And therefore:
		
		\begin{equation*}
		\begin{gathered}
		RF \cdot MSIDS \geq \rho \frac{1}{m} \sum_{v\in V}\frac{d(v)^2}{\rho (v)} 
		= \frac{1}{m|V|}  (\sum_{v\in V}\rho (v))(\sum_{v\in V}\frac{d(v)^2}{\rho (v)}) \geq\frac{4|E|^2}{m|V|}
		\end{gathered}
	\end{equation*}
		
	\end{proof}

	\section{Empirical results}\label{results}
	
	\subsection{Metrics calculations}
	
	In this subsection, we test the proposed method and consider the balance (formula \ref{eq:balance}), replication factor (\ref{eq:rf}), and the maximal sum of inner degrees (\ref{eq:msids}) of the resulting partitions.
	The datasets used here are graphs uk-2002, graph500-22, graph500-24, graph500-25, graph500-26. Graph parameters are provided in the table \ref{tab:graphs}. Partition algorithm Edge2DPartition is considered here as a baseline.

 \begin{table}[h]
  \caption{\label{tab:graphs}Graph datasets parameters: number of vertices and edges}
		\centering
		\begin{tabular}{c|c|c|c|c|c|}
			graphs & uk-2002 & graph500-22 & graph500-24 & graph500-25 & graph500-26 \\\hline
			$|V|$ &  18 520 486 & 2 396 657 & 8 870 942 & 17 062 472 &  32 804 978  \\\hline
            $|E|$ &  298 113 762 & 64 155 735 & 260 379 520 & 523 602 831 &  1 051 922 823
		\end{tabular}
	\end{table}
	
	Table \ref{tab:rfs} shows the replication factor values for the uk-2002 graph using different thresholds for the DBH-X method. The best value of the threshold is 80, so it is used further.
	
		\begin{table}[h]
  \caption{\label{tab:rfs}RF for different thresholds for uk-2002 graph (number of partitions 220).}
		\centering
		\begin{tabular}{c|c|c|c|c|c|c|}
			$\tau$ & 1 & 50 & 70 & 80 & 90 & 100\\\hline
			RF &  8.1619 & 8.1054 & 8.1014 & {\bf 8.1011 } &  8.1018 & 8.1038
		\end{tabular}
	\end{table}
	
	Table \ref{tab:metrics} displays the values of three metrics: balance, replication factor, and MSIDS for EdgePartition2D, standard DBH, and DBH-X with different parameters. A threshold is denoted as $\tau$. The number of partitions is 220. We observed that the version of DBH-X with a threshold and spread equal to 2 achieves the best replication factor, while the version with a spread of 20 achieves the best MSIDS. We can see that in this case, the replication factor with spread 2 is also improved, together with MSIDS, but higher spread values cause the replication factor to increase. It can also be checked that RF multiplied by MSIDS doesn't reach the lower bound from the theorem \ref{inequality}.

	\begin{table}[h]
    \caption{\label{tab:metrics}A metrics table. The number of partitions is 220}
		\centering
		\begin{tabular}{c|c|c|c|c|}
			graph & partition algorithm & balance & RF & MSIDS\\\hline
			uk-2002 & EdgePartition2D & 1.015 & 11.18 & 83.98 E6 \\
			& DBH & 1.006 & 8.16 & 84.16 E6\\
			& DBH-X $\tau=80$ & 1.007 & 8.105 & 85.64 E6\\
			& DBH-X  $\tau$=80\ spread = 2 & 1.027 & {\bf 7.37} & 57.92 E6\\
			& DBH-X  $\tau$=80\ spread = 5 & 1.003 & 10.65 & 24.64 E6\\
			& DBH-X  $\tau$=80\ spread = 10 & 1.024 & 11.53 &  27.68 E6\\
			& DBH-X  $\tau$=80\ spread = 20 & 1.022 & 15.71 & {\bf 13.96 E6}\\\hline 
			graph500-24 & EdgePartition2D & 1.061 & 11.77 & 8.41 E8\\
			& DBH & 1.050 & 8.16 & 21.08 E8 \\
			& DBH-X $\tau=500$ & 1.050 &  8.11 & 20.62 E8\\
			& DBH-X  $\tau$=500 spread = 2 & 1.036 & {\bf 6.68} & 14.96 E8\\
			& DBH-X  $\tau$=500 spread = 5 & 1.017 & 9.95 & 6.36 E8\\
			& DBH-X  $\tau$=500 spread = 10 & 1.017 & 9.57 & 7.99 E8\\
			& DBH-X  $\tau$=500 spread = 20 & 1.009 & 13.17 & {\bf 3.88 E8}\\
		\end{tabular}
	\end{table}
	
	\subsection{Performance improvement}
	
	The tests were conducted on a 4-node Apache Spark cluster with the following configuration: 4 nodes in total with one master node (driver) and 3 worker nodes. Each node had the following specifications: Kunpeng servers, 96 cores @2.6 MHz, 384 GiB RAM, and 10Gb/s network.

	We conducted testing using the GraphX framework and considered two algorithms: PageRank and Label Propagation. For the Label Propagation algorithm, we set the number of iterations to 10, while for PageRank, we ran it until convergence.
	
	The number of partitions is also important for runtime performance. Since this value influences metrics, it doesn't make sense to compare the replication factors of the two partition algorithms with different numbers of partitions. But for runtime, we chose the best parameters for each algorithm. In the figures \ref{fig:graphics1} and \ref{fig:graphics2}	 we provide graphics of runtime dependence on the number of partitions for EdgePartition2D and DBH-X partition algorithms.
	
	\begin{figure}[h]
		\centering
		\includegraphics[width=0.8\textwidth]{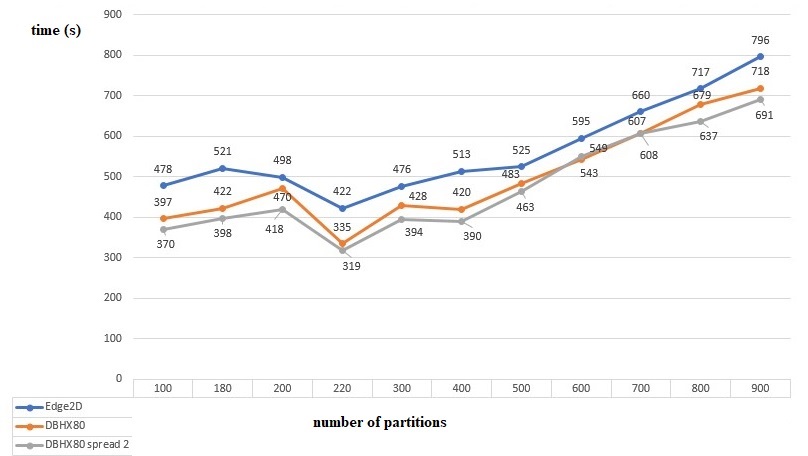}
		\caption{\label{fig:graphics1} Pagerank algorithm runtime using Edge2D and DBH-X partitions for graph uk-2002}
	\end{figure}
	
	\begin{figure}[h]
		\centering
		\includegraphics[width=0.8\textwidth]{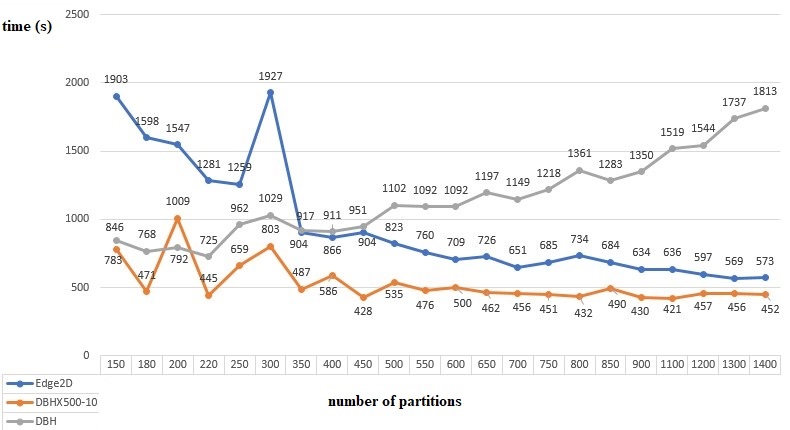}
		\caption{\label{fig:graphics2} Label propagation algorithm runtime using Edge2D, DBH and DBH-X partitions for graph graph500-24}
	\end{figure}
	
	We distinguish algorithm runtime and end-to-end time where graph partition process time is also included.  
	
	The improvement ratio in Tables \ref{tab:runtime} --- \ref{tab:lpa3} was calculated as the ratio of the runtime (end-to-end) for EdgePartition2D (Edge2D) to the runtime for DBH-X: 
	$$
	improv. = time_{Edge2D} / time_{DBH-X}.
	$$
	
	The algorithm time improvement is calculated similarly: 
	$$
	alg.time\ improv. = alg.time_{Edge2D} / alg.time_{DBH-X}.
	$$
	 
	In the table \ref{tab:lpa} we provide runtime and correspondent improvement for partition by 220 parts. The baseline here is EdgePartition2D with the same number of partitions. It can be seen that for the best runtime, we need some balance between the replication factor and MSIDS (data from table \ref{tab:metrics}).

 It's worth noting that we don't always need MSIDS minimization. For example, for
 Connected Components algorithm increasing spread parameter doesn't give runtime improvement, better runtime has DBH-X version with spread = 1.
	
	In the tables \ref{tab:lpa_btb} --- \ref{tab:lpa3} we provide a comparison between EdgePartition2D for the best number of partitions with DBH and DBH-X partitions with the best number of partitions. Besides graph500-24, we provide runtime also for graph500-25 and graph500-26.
	 
	\begin{table}[h!]
    \caption{\label{tab:runtime}PageRank algorithm run until convergence time (in seconds) for graph uk-2002, 220 partitions}
		\centering
		\begin{tabular}{c|c|c|c|c|c}
			partition & alg. & end-to-end & alg.time & end-to-end\\
			 & time & time & improv. & improv. \\\hline
			EdgePartition2D & 421.9 & 430.4 & - & - \\
			DBH-X $\tau$=80 & 334.7 & 358.5 & 1.26x & 1.20x\\
			DBH-X $\tau$=80 &  &  &  & \\
			spread=2 & 319.0 & 331.15 & 1.32x & 1.29x \\
			
		\end{tabular}
	\end{table}

	\begin{table}[h!]
    \caption{\label{tab:lpa}Label propagation algorithm runtime (in seconds) for graph graph500-24. num parts = 220}
		\centering
		\begin{tabular}{c|c|c|c|c|c|}
			partition & part. & alg. & end-to- & alg.time & end-to- \\
			 & time & time & -end time & improv. & -end impr.
			 \\\hline
			Edge2D & 7.2 & 1281.4 & 1288.6 & - & - \\\hline
			DBH & 27.8 & 724.7 & 752.5 & 1.76x & 1.71x \\\hline
			DBH-X $\tau$=500 & 16.7 & 770.0 & 786.7 & 1.66x & 1.64x \\
			spread = 1 & & & & & \\\hline
			DBH-X $\tau$=500 & 16.4 & 515.8 & 532.2 & 2.49x & 2.42x \\
			spread = 2 & & & & & \\\hline
			DBH-X $\tau$=500 & 16.5 &  548.6 & 565.1 & 2.33x & 2.28x\\
			spread = 5 & & & & & \\\hline
			DBH-X $\tau$=500 & 16.4 & 442.3 & 458.7 & {\bf 2.89x} & {\bf 2.81x}\\
			spread = 10 & & & & & \\\hline
			DBH-X $\tau$=500 & 16.5 & 531.5 & 548.0 & 2.41x & 2.35x\\
			spread = 20 & & & & & \\
		\end{tabular}
		
	\end{table}

	\begin{table}[h!]
    \caption{\label{tab:lpa_btb}Label propagation algorithm best runtime (in seconds) for graph graph500-24. Number of partitions = 220 for DBH, 1100 for DBH-X and 1300 for Edge2D}
	\centering
	\begin{tabular}{c|c|c|c|c|c|}
		partition & part. & alg. & end-to- & alg.time & end-to- \\
		& time & time & -end time & improv. & -end impr.
		\\\hline
		Edge2D & 8.0 & 568.9 & 577.0 & - & - \\\hline
		DBH & 27.8 & 724.7 & 752.5 & - & - \\\hline
		DBH-X $\tau$=500 & 21.2 & 420.7 & 442.0 & 1.35x & 1.30x\\
		spread = 10 & & & & 
	\end{tabular}
\end{table}
	
		\begin{table}[h!]
        \caption{\label{tab:lpa2}Label propagation algorithm runtime (in seconds) for graph graph500-25. Number of partitions = 220 for DBH, 1100 for DBH-X, and 1300 for Edge2D}
		\centering
		\begin{tabular}{c|c|c|c|c|c|}
			partition & part. & alg. & end-to- & alg.time & end-to- \\
			& time & time & -end time & improv. & -end impr.
			\\\hline
			Edge2D & 13.4 & 1447.1 & 1460.5 & - & - \\\hline
			DBH & 58.0 & 1729.7 & 1787.7 & - & - \\\hline
			DBH-X $\tau$=500 & 34.4 & 935.5 & 966.9 & 1.55x & 1.51x\\
			spread = 10 & & & & & \\		
		\end{tabular}
	\end{table}

	\begin{table}[h!]
    \caption{\label{tab:lpa3}Label propagation algorithm runtime (in seconds) for graph graph500-26. Number of partitions = 220 for DBH, 1100 for DBH-X, and 1300 for Edge2D}
    \centering
	\begin{tabular}{c|c|c|c|c|c|}
		partition & part. & alg. & end-to- & alg.time & end-to- \\
		& time & time & -end time & improv. & -end impr.
		\\\hline
		Edge2D & 23.7 & 3567.5 & 3591.2 & - & - \\\hline
		DBH & 116.3 & 4591.4 & 4707.7  & - & - \\\hline
		
		DBH-X $\tau$=500 & 63.6 & 2165.9 & 2229.5 & 1.65x & 1.61x\\
		spread = 10 & & & & & \\
	
	\end{tabular}
\end{table}

	\section{Conclusion}\label{conclusion}
	
	Our paper presents DBH-X, a novel vertex-cut graph partitioning method that extends the well-known DBH method. We also introduce a novel metric for graph partitioning and demonstrate its relation with the replication factor. Our study shows that DBH-X strikes a balance between this new metric and the replication factor, leading to improved performance compared to DBH in practical applications. We believe graph partitioning methods that utilize graph degrees have the potential for further improvement, and we plan to explore this in future research.
	
 \bibliographystyle{plain}
  \bibliography{bibliography}

\end{document}